\documentclass[12pt,tightenlines,eqsecnum,floats,aps,amsmath,amssymb,nofootinbib,prd,showpacs]{revtex4}

\usepackage{setspace}
\usepackage{amsmath,amssymb,amsfonts,amsthm}
\usepackage{graphicx}
\usepackage{enumerate} 
\usepackage{colordvi} 

\def\be{\begin{equation}}
\def\ee{\end{equation}}
\def\beq{\begin{eqnarray}}
\def\eeq{\end{eqnarray}}
\newtheorem{thm}{Theorem}[section]
\newtheorem{lem}[thm]{Lemma}

\begin{document}
\title{Static Solutions for $4^{th}$ Order Gravity}

\author{William Nelson\footnote{nelson@gravity.psu.edu} }
\affiliation{Institute of Gravitation and the Cosmos, Penn State University, State College, PA 16801, U.S.A. }

\begin{abstract}
The Lichnerowicz and Israel theorems are extended to higher order theories of gravity. In particular it
is shown that Schwarzschild is the {\it unique} spherically symmetric, static, asymptotically flat, black-hole solution,
provided the spatial curvature is less than the quantum gravity scale {\it outside} the horizon. It is then
shown that in the presence of matter (satisfying certain positivity requirements), the only static and
asymptotically flat solutions of General Relativity that are also solutions of higher order gravity are the
vacuum solutions. 
\end{abstract}
\maketitle

%
%

\section{Introduction} \label{sec:intro}

Whilst General Relativity (GR) remains the most successful classical theory we have, one may
question whether it is possible for there to be corrections to it. Indeed we expect GR to break down
as we approach the Planck scale where the quantum theory of gravity (whatever it is) will become
dominant. In this sense GR is not a full physical theory, only an excellent approximation to some
(presumably) complicated underlying theory, with the approximation getting better and better
as the scales we consider are further and further from the Planck scale\footnote{It is important
to note that here the Planck scale can include both very small (UV) scales {\it and} very large (IR)
scales (see for example\cite{Adams:2006sv}).}.

Over the past century there have been many important theorems proved about the
solutions and structure of GR and an important question is whether these theorems are also valid in
theories that deviate from GR at some scale. In this paper we will look in particular at two theorems,
firstly the Lichnerowicz theorem~\cite{lichnerowicz}, which tells us that the only static, asymptotically flat,
geodesically complete, vacuum, solution to Einstein's equations is flat space-time. We then consider
the Israel theorem~\cite{Israel:1967wq} which demonstrates that the only static, asymptotically flat,
vacuum space-time, which contains past
and future event horizons (that intersect on a surface that is topologically $S^2$) is given by 
the Schwarzschild metric\footnote{The Israel theorem also shows that the Reissner-Nordstr\"{o}m metric must be
the solution if the vacuum is replaced by an electromagnetic field.}.
This `no-hair' theorem
is a striking result in classical GR and extensions of it to more general gravity theories, in
particular as we approach the scales on which GR is expected to be violated, would provide us
with insight into the transition from GR to full quantum gravity.

The action for GR is the well-known Einstein-Hilbert action which is linear in the Ricci
scalar $R$ (here for simplicity we neglect the cosmological constant).
A natural extension of this is to include terms of higher powers of curvature 
invariants to find,
\be\label{eq:action_gen}
 {\cal S}_{\rm Grav.} = \int {\rm d}^4 x \sqrt{-g}\left( \frac{\gamma}{\kappa^2} R
 -\alpha_0 \hbar R^{abcd}R_{abcd} - \alpha \hbar R^{ab}R_{ab} + \beta \hbar R^2 + {\cal O}\left( R^3\right) \right)~,
\ee
where we have written only terms that are only $2^{\rm nd}$ order in the curvature. Through-out
we will use Latin indices to label space-time components $a,b,\dots = 0,\dots,3$, our signature is
$\left( -,+,+,+\right)$ and our conventions for curvature are $R_{abc}^{~~~d} = \partial_b \Gamma^d_{~ac}
+\dots$ and $R_{ab} = R^c_{~acb}$. In Eq.~(\ref{eq:action_gen}) the coefficients $\alpha_0$, $\alpha$,
$\beta$ and $\gamma$ are dimensionless number and $\kappa^2 = 32\pi G$,
so GR is recovered simply by setting the coefficients of the higher order terms to zero ($\alpha_0=\alpha=\beta=0$)
and taking $\gamma = 2$. Also note that the presence of the $\hbar$ in Eq.~(\ref{eq:action_gen}) is only due to
dimensions i.e.\ this remains a classical theory. If these correction terms are motivated from
some quantum gravity theory, then we may expect the coefficients $\alpha_0$, $\alpha$ and $\beta$ to be
order unity, however in general they can take any value.

In fact, in $4$-dimensions, the Gauss-Bonnet combination of second order curvature invariants
integrates to a topological invariant, i.e.\
\be
 R^{abcd}R_{abcd} - 4R^{ab}R_{ab} + R^2 = {\rm total\ divergence}~,
\ee
which we can use to eliminate the $R^{abcd}R_{abcd}$ term from Eq.~(\ref{eq:action_gen}). Thus, up to
second order in the curvature, the most general correction to GR comes from the action,
\be\label{eq:action}
 {\cal S}_{\rm Grav.} = \int {\rm d}^4 x \sqrt{-g}\left( \frac{\gamma}{\kappa^2} R
 - \alpha\hbar R^{ab}R_{ab} + \beta \hbar R^2  \right)~.
\ee
This is the theory that we will work with and in the following we refer to it as
$4^{\rm th}$ order gravity, since the equations of motion involve, at most, $4^{\rm th}$
derivatives of the metric. Historically, theories of this type, with $\gamma=0$ were
considered for many reasons, not least because of their improved behaviour under renormalisation,
however it was shown that without a term linear in $R$, the theory (for $\alpha=0$) does not 
couple to matter correctly~\cite{Pechlaner:1987nt}. The theory, with general $\gamma$
was first considered in~\cite{Gregory:1947zz} and later in~\cite{Stelle:1977ry}, where
a perturbative analysis around the Schwarzschild solution was performed, whilst the initial value
formulation of the theory was given in~\cite{Noakes:1983xd}. The consequence for black holes in
higher dimensions has been consider in~\cite{Frolov:2009qu} and the cosmological implications
of $4^{\rm th}$ order gravity were investigated in~\cite{Starobinsky2,Starobinsky1,Schmidt:2006jt,Macrae:1981ic},
with the extension to include
(tensor) perturbations appearing only recently~\cite{Nelson:1}.
The special case in which  the correction terms are conformal (when $\alpha = 3\beta$, which gives
the Bach-Einstein equations) has received considerable attention (see for example~\cite{Schmidt:1984bf})
and has recently appeared in the context of non-commutative geometry~\cite{Chamseddine:2006ep}, with
much work going into restricting the parameters of this form of the theory via cosmological and astrophysical
observations~\cite{Nelson:2008uy,Nelson:2009wr,Nelson:2010rt,Nelson:2010ru,Kolodrubetz:2010xi}. More
generally, corrections to GR of this type are expected from String Theory (see for example~\cite{Johnson:book}),
renormalization group techniques~\cite{Reuter:2001ag} and Loop Quantum Gravity~\cite{Bojowald:2006ww}.

In addition there has been a great deal of phenomenological work on the cosmological consequences of higher order gravity
theories, usually in the form of $f\left( R\right)$ theories (see for example ~\cite{Sotiriou:2008rp,delaCruzDombriz:2009et,Nojiri:2006ri}).
 Despite this it is important to note
that theories defined by Eq.~(\ref{eq:action}) cannot represent true physical theories for arbitrary coefficients 
due to the presence of ghosts~\cite{Stelle:1977ry,Kawai:1998ab,DeFelice:2006pg}. Not only are ghost a problem for the quantisation of the theory, they
also lead to a break-down of causality at the classical level. Note however that with certain restrictions on the
coefficients (e.g.\ $\alpha=0$ and $\beta<0$), all the difficulties associated with these ghosts can be eliminated.
Here we will take the (conservative) view that, just as for GR,  $4^{\rm th}$ order gravity should be
considered only as an approximation to the underlying, full theory (which is presumable free of ghosts etc.).

In this paper we shall demonstrate three key results, the first two: an extension of the Lichnerowicz
(Section~(\ref{sec:lich_thm})) and
Israel (Section~(\ref{sec:israel})) theorems to $4^{\rm th}$ order gravity (for space-times satisfying some restriction
on the magnitude of the spatial curvature) suffice to demonstrate that
Schwarzschild is the {\it unique}, spherically symmetric, static, asymptotically
flat, vacuum solution to the theory\footnote{This result holds for black-holes in which the spatial curvature
outside the horizon is small relative to the scale set between $\gamma \kappa^{-2}$ and $\hbar \alpha$. If
these corrections come from some quantum gravity theory, this restriction is essentially that the spatial
curvature be small compared to the quantum gravity scale.}. The final result (Section~(\ref{sec:non_vac}))
considers non-vacuum solutions to the theory and demonstrates that for static, asymptotically flat 
space-times with matter satisfying certain positivity requirements, it is {\it only} the vacuum solutions of GR
and $4^{\rm th}$ order gravity that agree; all other solutions of GR fail to be solutions of $4^{\rm th}$ order gravity.

%
%
\section{Static solutions of the equations of motion}\label{sec:1}
Varying the action given in Eq.~(\ref{eq:action}) with respect to the metric, one finds~\cite{Stelle:1977ry},
\beq\label{eq:EoM}
 H_{ab} &=& \left( \alpha - 2\beta\right) R_{;a;b} - \alpha \Box R_{ab} - \frac{1}{2}\left( \alpha -4\beta\right)
g_{ab}\Box R + 2\alpha R^{cd}R_{acbd} \nonumber \\
&&- 2\beta R R_{ab} - \frac{1}{2} g_{ab} \left( \alpha R^{cd}R_{cd}
- \beta R^2\right) - \frac{\gamma}{\hbar \kappa^{2}}\left( R_{ab} - \frac{1}{2} g_{ab}R\right) = -\frac{1}{2\hbar}T_{ab}~.
\eeq
Considering the vacuum, the trace part gives the Klein-Gordon equation for
the Ricci scalar,
\be\label{eq:trace}
 \left( \alpha - 3\beta\right)\left( \Box - \frac{\gamma \kappa^{-2}}{2\left( \alpha - 3\beta\right)} \right) R =0~.
\ee
It is well known that there are no static, asymptotically constant solutions of the
Klein-Gordon equation for a scalar field (see for example~\cite{Schmidt}), however since the method used to prove this result
will be the one that is later extended to the full equation of motion, Eq.~(\ref{eq:EoM}),
we give it explicitly here.

\begin{thm}\label{thm:p1}
If the space-time is static (${\cal L}_{t} g_{ab}=0$) and $D_aR\rightarrow 0$
sufficiently fast at infinity, then $\left( \Box - m^2 \right) R=0$, for $m\neq0 \in \mathbb{R}$ implies $R=0$.
If the space-time is static and $R\rightarrow 0$ sufficiently fast at infinity
the equation $\Box R=0$ implies $R=0$.
\end{thm}
\begin{proof}
We decompose the metric into components parallel and perpendicular to the constant time space-like hyper-surfaces ${\cal S}$ via
$g_{ab} = h_{ab} - \frac{1}{\lambda} t_a t_b$, where $t_a$ is the time-like Killing field, which has norm $t^at_a = -\lambda$
and $h_{ab}$ is spatial metric on ${\cal S}$.
Note the following useful identities
\be
 {\cal L}_t R= t^b\nabla_b R =0~,
\ee
\be\label{eq:0.1}
 \nabla_a t_b = \frac{1}{2\lambda} \left( t_b\nabla_a\lambda - t_a \nabla_b\lambda\right)~,
\ee
and
\be\label{eq:0.2}
 t^a\nabla_a \lambda = 0~.
\ee

Using these one can readily show that,
\beq
 g^{ab}\nabla_a \nabla_b R - m^2 R = D^a D_a R+ \frac{1}{2\lambda} \left( D^a \lambda\right)\left( D_a R\right) - m^2 R = 0~,
\eeq
where $D_a$ is the spatial covariant derivative compatible with $h_{ab}$.
Multiplying this equation by $\lambda^{1/2}R$ and integrating it over the entire spatial slice, one finds
\beq
&& \int_{\cal S} \sqrt{h} {\rm d}^2x\left[ \lambda^{1/2} RD^a D_a R + \frac{1}{2} \lambda^{-1/2} R\left( D^a\lambda \right)\left( D_a
R\right) - m^2 \lambda^{1/2} R^2 \right] =0~.
\eeq
Integrating the first term by parts this becomes,
\beq\label{eq:pf1}
 \int_{\cal S} \sqrt{h}{\rm d}^3x\left[ D^a\left( \lambda^{1/2} RD_a R\right) - \lambda^{1/2} \left( D^a R\right)\left( D_a R\right)
-m^2 \lambda^{1/2} R^2\right] = 0~.
\eeq
If the spatial slice is asymptotically constant i.e.\ $D_a R \rightarrow 0$ sufficiently fast at infinity,
then the boundary term vanishes. The second and third
terms in the integrand of Eq.~(\ref{eq:pf1}) are manifestly negative definite and hence each term must vanish at every spatial point.
If $m\neq0$ this implies $R=0$, for the case of $m=0$ this gives $D_aR=0$ and hence we additionally require
that $R\rightarrow 0$ at infinity  in order to find $R=0$.
\end{proof}
Note that here we require only that $D_aR\rightarrow 0$ (or $R\rightarrow 0$) sufficiently fast at infinity, which is
satisfied by asymptotically constant (or flat) space-times, but is in general a weaker condition. In the following
we will consider asymptotically constant (or flat) space-times since they are of more interest,
however all the results hold also for the weaker condition given above.

The basic method of this proof is the following. Decompose the metric onto (and perpendicular to)
constant time space-like hyper-surfaces and use this to write the $4$-dimensional covariant derivatives as spatial, $3$-dimensional
covariant derivatives plus correction terms. Multiply the resulting equation by a suitable factor and integrate over the
space-like hyper-surface. Prove that, up to boundary terms, the integrand has a definite sign and hence vanishes at each spatial
point. This is precisely the same method that we will use to prove that, with some restrictions on the magnitude
of the spatial curvature on ${\cal S}$, there are no static, asymptotically constant
solutions to the trace-free part of Eq.~(\ref{eq:EoM}), which, using theorem~(\ref{thm:p1}) can be written as,
\beq\label{eq:trace_free}
\hbar \alpha
 \left( \Box R_{ab} - \frac{\gamma}{\hbar\alpha \kappa^{2}} R_{ab} + \frac{1}{2}g_{ab} R^{cd}R_{cd} - 2R^{cd}R_{acbd}\right) =0~,
\eeq
for $\alpha-3\beta \geq 0$.
Since Eq.~(\ref{eq:trace_free}) is rather more complicated than Eq.~(\ref{eq:trace}), we proceed by first
proving a series of lemmas~\footnote{A similar line of argument has used in~\cite{Schmidt}. }.

\begin{lem}\label{lem:1}
If the space-time is static and asymptotically constant,
then $\left( \Box - m^2 \right) R_{ab}=0$, for $m\neq 0\in \mathbb{R}$ implies $R_{ab}=0$. 
If the space-time is static and asymptotically flat,
then $\Box  R_{ab}=0$, implies $R_{ab}=0$. 
\end{lem}

\begin{proof}
We begin by noting several identities, which are proved in Appendix~\ref{app}.
\beq\label{eq:1}
t^a \nabla _a t^b \nabla_b R_{cd} = 2\left[ \left( R_{a(c} \nabla_{d)} t^b\right)\left( \nabla_bt^a\right) +
R_{ab}\left( \nabla_c t^b\right) \left( \nabla_d t^a\right)\right]
\eeq
\beq\label{eq:2}
 \left( \nabla _d t^c\right) \left( \nabla_at^d\right) = \frac{-1}{4\lambda^2} t^c t_a \left( \nabla_d\lambda\right)
\left( \nabla^d \lambda\right) + \frac{1}{4\lambda} \left( \nabla^c\lambda\right) \left( \nabla_a\lambda\right)~,
\eeq
\beq\label{eq:2.25}
 \frac{1}{2\lambda^2}h^{ab}\left( \nabla_a\lambda\right) t^ct^d\nabla_bR_{cd} = \frac{1}{2\lambda} 
 \left( \nabla^a\lambda\right)\left(\nabla_a ^{(3)}R\right)~,
\eeq
and
\beq\label{eq:2.75}
 \frac{1}{\lambda}t^ct^dh^{ab}\nabla_a h^e_b \nabla_e R_{cd} = D^aD_a ^{(3)}R~,
\eeq
where $^{(3)}R$ is the $3$-dimensional Ricci scalar formed from the spatial metric $h_{ab}$.
Now, decomposing $g_{ab} = h_{ab} - \frac{1}{\lambda}t_a t_b$ one finds,
\be\label{eq:3}
 \left( \Box - m^2 \right) R_{ab} = h^{ab}\nabla_a h^e_b \nabla_e R_{cd} - \frac{1}{\lambda} t^a t^b \nabla_a \nabla_b R_{cd}
-m^2 R_{cd} = 0~.
\ee
Noting the staticity of $R_{ab}$,
\be
 {\cal L}_t R_{ab} \equiv t^a \nabla_a R_{cd} + R_{ad} \nabla_ct^a + R_{ca} \nabla_dt^a=0~,
\ee
and using Eq.~(\ref{eq:1}) and Eq.~(\ref{eq:2}), one can show (see Appendix~\ref{app})
that Eq.~(\ref{eq:3}) becomes,
\beq\label{eq:3.5}
&& h^{ab} \nabla_a h^e_b\nabla_e R_{cd} + \frac{1}{2\lambda}h^{ab}\left( \nabla_a \lambda\right)
\left( \nabla_bR_{cd}\right)
+ \frac{1}{4\lambda^3} \left[ t_ct^a\left( \nabla_b\lambda\right)
\left( \nabla ^b\lambda\right) - \lambda \left( \nabla_c\lambda\right) \left( \nabla^a\lambda\right)
\right] R_{ad} \nonumber \\
&&-\frac{1}{2\lambda^2}\left[ t^a\nabla_d\lambda - t_d\nabla^a\lambda\right]\left[ t^b\nabla_c\lambda
-t_c\nabla^d\lambda\right] R_{ba} 
+ \frac{1}{4\lambda^3} \left[ t_dt^a\left( \nabla_b\lambda\right)
\left( \nabla ^b\lambda\right) - \lambda \left( \nabla_d\lambda\right) \left( \nabla^a\lambda\right)
\right] R_{ac}\nonumber \\
&& -m^2R_{cd} = 0~.
\eeq
We now project this equation onto the spatial slice, ${\cal S}$ with $h^c_{e}h^d_{f}$ and
define the projection of $R_{ab}$ onto ${\cal S}$ as $\bar{R}_{ab} \equiv h_a^c h_b^dR_{cd}$
to find,
\beq\label{eq:4}
&& D^a D_a \bar{R}_{ef} + \frac{1}{2\lambda} \left( D_a\lambda\right) \left(
D^a\bar{R}_{ef}\right) - \frac{1}{4\lambda^2}\left( D_{e}\lambda\right)
\left( D^a\lambda\right) \bar{R}_{af} - \frac{1}{2\lambda^3} \left( D_{f}
\lambda\right)\left( D_{e}\lambda\right) t^at^bR_{ab}\nonumber \\
&& - \frac{1}{4\lambda^2}
\left( D^a\lambda\right) \left( D_{f} \lambda\right) \bar{R}_{ea}
-m^2 \bar{R}_{ef}=0~.
\eeq
Noting that the extrinsic curvature of the spatial slice vanishes (as can be shown by
direct calculation) the Codacci and Gauss equations reduce to,
\beq\label{eq:Codazzi}
 h_a^b t^c R_{bc} &=& 0~,\\
\label{eq:Gauss}
 \frac{1}{\lambda} t^a t^b R_{ab} &=& ^{(3)}R~,
\eeq
where we have use Theorem~(\ref{thm:p1}) to set $R=0$. Using these relations in Eq.~(\ref{eq:4})
we find,
\beq
&& D^a D_a \bar{R}_{cd} + \frac{1}{a\lambda} \left( D_a\lambda\right)\left( D^a\bar{R}_{cd}\right)
 -\frac{1}{4\lambda^2}\left( D_c\lambda\right) \left( D^a\lambda\right) \bar{R}_{ad}
-\frac{1}{2\lambda^2} \left( D_d\lambda \right) \left( D_c\lambda\right) ^{(3)}R\nonumber \\
&&-\frac{1}{4\lambda^2} \left( D^a\lambda\right)\left( D_d\lambda\right) \bar{R}_{ca} - m^2\bar{R}_{cd}=0~.
\eeq
Multiplying this by $\lambda^{1/2} \bar{R}^{cd}$ and integrating over the $3$-dimensional spatial slice gives,
\beq\label{eq:5}
&& \int_{\cal S} \sqrt{h} {\rm d}^3 x\Biggl[  D^a\left( \lambda^{1/2}\bar{R}^{cd} D_a \bar{R}_{cd} \right)
 - \lambda^{1/2}\left( D^a\bar{R}^{cd}\right)\left( D_a\bar{R}_{cd}\right) 
-\frac{1}{2} \lambda^{-3/2} \bar{R}^{cd}\bar{R}_{ad} \left( D_c\lambda\right)\left( D^a\lambda\right) \nonumber \\
&&-\frac{1}{2} \lambda^{-3/2} \left.\right.^{(3)}R\bar{R}^{cd} \left(D_d\lambda\right)\left( D_c\lambda\right)
-\lambda^{1/2}m^2 \bar{R}_{cd}\bar{R}^{cd}\Biggr]=0~.
\eeq

We now project Eq.~(\ref{eq:3.5}) perpendicular to the space-like hyper-surface, with $t^b t^c$ and again
use the Codacci and Gauss equations (Eq.~(\ref{eq:Codazzi}) and Eq.~(\ref{eq:Gauss}) respectively) to find,
\beq
&& t^c t^d h^{ab}\nabla_ah^e_b\nabla_eR_{cd} + \frac{1}{2\lambda} t^c t^dh^{ab}\left( \nabla_a\lambda\right)\nabla_b
R_{cd} - \frac{1}{2\lambda} \left( \nabla_b\lambda\right)\left( \nabla^b\lambda\right) ^{(3)}R 
\nonumber \\
&&- \frac{1}{2\lambda}\left( \nabla^a\lambda\right)\left( \nabla^b\lambda\right) R_{ab} 
- \lambda m^2 \left.\right. ^{(3)}R =0~.
\eeq
Multiplying this by $\lambda^{-1/2}\left.\right.^{(3)}R$, integrating over the $3$-dimensional
hyper-surface and using the identities given in Eq.~(\ref{eq:2.25}) and Eq.~(\ref{eq:2.75}), gives,
\beq\label{eq:6}
&& \int_{\cal S} \sqrt{h}{\rm d}^3x\Biggl[ D^a\left( \lambda^{1/2} \left.\right.^{(3)}R D_a^{(3)}R\right)
- \lambda^{1/2}\left( D^a\left.\right.^{(3)}R\right)\left( D_a\left.\right.^{(3)}R\right)
\nonumber \\
&&-\frac{1}{2}\lambda^{-3/2} \left( D_a\lambda\right)\left( D^a\lambda\right)\left(
^{(3)}R\right)^2 - \frac{1}{2} \lambda^{-3/2} \left( D^a\lambda\right)\left( D^b\lambda\right)
\bar{R}_{ab}~ ^{(3)}R - \lambda^{1/2}m^2\left( ^{(3)}R\right)^2\Biggr]=0~.\nonumber \\
\eeq

Finally, adding Eq.~(\ref{eq:5}) and Eq.~(\ref{eq:6}) one finds,
\beq\label{eq:7}
&& \int_{\cal S} \sqrt{h} {\rm d}^3 x\Biggl[ D^a\left( \lambda^{1/2} \bar{R}^{cd} D_a \bar{R}_{cd}
+\lambda^{1/2} \left.\right.^{(3)}RD_a~^{(3)}R\right)
-\lambda^{1/2}\left( D^a\bar{R}^{cd}\right)\left( D_a\bar{R}_{cd}\right)
\nonumber \\
&&-\lambda^{1/2}\left( D^a\left.\right.^{(3)}R\right)\left( D_a~^{(3)}R\right)
-\frac{1}{2}\lambda^{-3/2} \left( \bar{R}^{cd} D_c\lambda + ^{(3)}R D^d\lambda\right)
\left(  \bar{R}_{ad} D^a\lambda + ^{(3)}R D_d\lambda\right)
\nonumber \\
&&-\lambda^{1/2} m^2 \left( \bar{R}^{cd} \bar{R}_{cd} + \left(^{(3)}R\right)^2\right)\Biggr]=0~.
\eeq
This is sufficient to prove the desired result, however there is a simplification that can be made
by noting that the contracted Bianchi identity $\nabla^a\left( R_{ab} - \frac{1}{2}g_{ab}R\right)=0$
implies that
\be\label{eq:contracted_bianchi}
 \bar{R}_{dc}D^d\lambda + ^{(3)}RD_c\lambda = -2\lambda D^a\bar{R}_{ac}~.
\ee
Substituting this into Eq.~(\ref{eq:7}) explicitly removes the derivatives of $\lambda$ from the integrand and ensures that $\lambda$
only appears with positive powers. This will be crucial when we consider space-times with (null) horizons.

In any case, asymptotic constancy i.e. $D_a R_{cd} \rightarrow 0$ at infinity, ensures that the 
boundary terms in Eq.~(\ref{eq:7}) vanish. All the remaining terms are negative definite and hence
each term must independently vanish at every spatial point. For $m\neq 0$ this implies
$^{(3)}R=0$ and $\bar{R}_{cd}=0$, which together with theorem~(\ref{thm:p1}) and Eq.~(\ref{eq:Codazzi})
implies $R_{ab}=0$. For the $m=0$ case we find that $D_a^{(3)}R=0$ and $D_a\bar{R}_{cd}=0$ and hence  
additional require asymptotic flatness in order to find $R_{ab}=0$.
 
\end{proof}

This result is easily extended to include the third term in Eq.~(\ref{eq:trace_free}) by using the 
fact that, by Theorem~(\ref{thm:p1}) we can set $R=0$.
\begin{lem}\label{lem:2}
If the space-time static and asymptotically constant, then the pair of equations
$R=0$ and  $\left( \Box - m_1^2\right) R_{ab} + m_2g_{ab}R^{cd}R_{cd}=0$,
for $m_1\neq0 \in \mathbb{R}$ and $m_2 \in \mathbb{R}$ imply $R_{ab}=0$, whilst
for $m_1=0$, additionally requiring the space-time to be asymptotically flat implies $R_{ab}=0$.
\end{lem}
\begin{proof}
From the equation $\left( \Box - m_1^2\right) R_{ab} + m_2^2g_{ab}R^{cd}R_{cd}=0$ it
follows that,
\beq\label{eq:8}
 \left[ \lambda^{1/2} \bar{R}^{cd}h_c^a h_d^b + \lambda^{-1/2}\left.\right. ^{(3)}Rt^at^b\right]
\left[ \left( \Box - m_1^2\right) R_{ab} + m_2^2g_{ab}R^{cd}R_{cd}\right] =0~.
\eeq
Recalling that $t_at^a=-\lambda$ and noting that $h^{ab}\bar{R}_{ab} = ^{(3)}R$ 
we see that Eq.~(\ref{eq:8}) reduces to,
\be\label{eq:9}
 \left[ \lambda^{1/2} \bar{R}^{cd}h_c^a h_d^b + \lambda^{-1/2}\left.\right. ^{(3)}Rt^at^b\right]
\left( \Box - m_1^2\right) R_{ab} = 0~,
\ee
which is independent of the value of $m_2$. By Lemma~(\ref{lem:1}), for $m_1\neq 0$, this implies $R_{ab}=0$,
provided the space-time is static and asymptotically constant, whilst for $m_1 = 0$ and the 
space-time being static and asymptotically flat implies $R_{ab}=0$.
Indeed integrating Eq.~(\ref{eq:9}) over the spatial slice exactly gives  Eq.~(\ref{eq:7}).
\end{proof}

The final step that is required in order to allow us to consider Eq.~(\ref{eq:trace_free}) is
to extend the above Lemmas to include a term of the form $R^{ab}R_{acbd}$. To do this
note that by using the definition of the Riemann tensor,
\be
 \left( \nabla_a\nabla_b - \nabla_b\nabla_a\right) R_{cd} = R_{abc}^{~~~e}R_{ed} + R_{abd}^{~~~e}R_{ce}~,
\ee
and the fact that the contracted Bianchi identity, for $R=0$, gives,
\be
 \nabla_cR_b^{~c} = \frac{1}{2} \nabla_bR = 0~,
\ee
we can write (for $R=0$),
\be\label{eq:10}
 R^{cd} R_{acbd} = - \nabla_c \nabla_a R_b^{~c} + R_{ac}R_b^{~c}~.
\ee
Thus we see that this term is in fact closely related to $\Box R_{ab}$ and can be dealt with in a similar manner.
\begin{lem}\label{lem:3}
 For a static space-time the following relation holds:
\beq\label{eq:lem_proof}
&& \int \sqrt{h}{\rm d}^3x\left[ \left( \lambda^{1/2} \bar{R}^{ef} h^a_e h^b_f
 + \lambda^{-1/2} \left.\right.^{(3)}R t^at^b \right)R_{acbd}R^{cd} \right] = \nonumber \\
&& \int \sqrt{h} {\rm d}^3 x \Biggl[ 
D^a\left( -\lambda^{1/2} \bar{R}^{ef} D_e \bar{R}{af} + \lambda^{1/2} \left.\right.^{(3)}R
D^b\bar{R}_{ab}\right)
+\lambda^{1/2} \left( D_a \bar{R}_{ef}\right)\left( D^e \bar{R}^{af}\right) \nonumber \\
&&+ \lambda^{1/2} \left( D^a \bar{R}_{ac} \right) \left( D_b \bar{R}^{bc}\right)
-2\lambda^{1/2} \left( D^c \bar{R}_{ca}\right)\left( D^a \left.\right.^{(3)}R\right)
+\lambda^{1/2} \left[ \bar{R}^a_{~b}\bar{R}^b_{~c}\bar{R}^c_{~a} - \left( ^{(3)}R\right)^3\right]
\Biggr]~.\nonumber \\
\eeq
\end{lem}
\begin{proof}
To show this we project Eq.~(\ref{eq:10}) onto the spatial surface, ${\cal S}$, with $h^a_{~e} h^b_{~f}$ and use the identity
\beq
 h^a_{~e} h^b_{~f} \nabla_c \nabla_a R_b^{~c} &=&  h^a_{~e} h^b_{~f} \left( h^{cd} - \frac{1}{\lambda}t^ct^d\right)
\nabla_c \nabla_a R_{bd}\nonumber \\
&=& D_a D_e \bar{R}^a_{~f} - \frac{1}{4\lambda^2} \left( D_e\lambda\right)\left(D_f\lambda\right) ^{(3)}R
-\frac{1}{4\lambda^2} \left( D_e\lambda\right)\left( D^d\lambda\right) \bar{R}_{fd}
\nonumber \\
&&+ \frac{1}{2\lambda} \left( D_f\lambda\right) \left( D_e ^{(3)}R\right) + \frac{1}{2\lambda}
\left( D^a\lambda\right) \left( D_e\bar{R}_{fa}\right)~,
\eeq
where the second equality comes from the tedious but straight forward algebra coming from
 bringing the $h^b_{~f}t^d$ inside both covariant derivatives and
repeatedly using Eq.~(\ref{eq:0.1}), Eq.~(\ref{eq:0.2}), Eq.~(\ref{eq:Codazzi}) and Eq.~(\ref{eq:Gauss}).

Substituting this expression into Eq.~(\ref{eq:10}), multiplying by $\lambda^{1/2} \bar{R}^{ef}$ and
integrating over ${\cal S}$, one finds,
\beq\label{eq:10.5}
 && \int_{\cal S} \sqrt{h}{\rm d}^3x\left[ \lambda^{1/2} \bar{R}^{ef} h^a_e h^b_f R_{acbd}R^{cd}\right] =
\nonumber \\
&&
\int_{\cal S} \sqrt{h}{\rm d}^3x\Biggl[ D^a\left( -\lambda^{1/2} \bar{R}^{ef}D_e \bar{R}_{af}\right)
+\lambda^{1/2} \left( D_a\bar{R}_{ef}\right)\left( D^e\bar{R}^{af}\right)
+ \frac{1}{4}\lambda^{-3/2} \left.\right. ^{(3)}R \bar{R}^{ef} \left( D_e\lambda\right)\left( D_f\lambda\right)
 \nonumber \\
&&
+\frac{1}{4}\lambda^{-3/2} \bar{R}^{ef}\bar{R}_{fd} \left( D_e \lambda\right) \left( D^d\lambda\right)
-\frac{1}{2}\lambda^{-1/2} \bar{R}^{ef} \left( D_f\lambda\right)\left( D_e^{(3)}R\right) + \lambda^{1/2} \bar{R}^a_{~b}
\bar{R}^b_{~c} \bar{R}^c_{~a}\Biggr]~.
\eeq
Similarly we project Eq.~(\ref{eq:10}) along $t^at^b$ to find,
\beq\label{eq:11}
 t^at^b R^{cd}R_{acbd} &=& -t^at^b\nabla_c\nabla_a R_b^{~c} + R_{ac}R_b^{~c}t^at^b~,\nonumber \\
&=& \frac{1}{2} \left[\left( D_a\bar{R}^a_{~b}\right)\left( D^b\lambda\right) + \bar{R}^{ab} D_aD_b\lambda
+^{(3)}RD_aD^a\lambda\right] - \lambda\left[ ^{(3)}R\right]^2~,
\eeq
where the second equality follows from the staticity of the space-time and again, repeated use of
Eq.~(\ref{eq:0.1}), Eq.~(\ref{eq:0.2}), Eq.(\ref{eq:Codazzi}) and Eq.~(\ref{eq:Gauss}).

Multiplying Eq.~(\ref{eq:11}) by $\lambda^{-1/2}\left.\right.^{(3)}R$ and integrating over ${\cal S}$, one finds,
\beq\label{eq:12}
&& \int_{\cal S} \sqrt{h}{\rm d}^3 x\left[ \lambda ^{-1/2}\left.\right.^{(3)} R t^a t^b R^cd R_{acbd}\right] =
\nonumber \\
&&
\int_{\cal S} \sqrt{h} {\rm d}^3 x\Biggl[ D^a \left( \frac{1}{2} \lambda^{-1/2} \left.\right.^{(3)}R \bar{R}_{ab}
D^b\lambda + \frac{1}{2} \lambda^{-1/2} \left.\right.^{(3)}R ^{(3)}R D^a\lambda\right) \nonumber \\
&&
+ \frac{1}{4} \lambda^{-3/2} \left.\right.^{(3)} \bar{R}^{ab}\left( D_a \lambda\right)\left(
D_b\lambda\right) + \frac{1}{4} \lambda^{-3/2} \left. \right.^{(3)}R ^{(3)}R \left( D^a\lambda\right)
\left(D_a\lambda\right) \nonumber \\
&&
- \frac{1}{2} \lambda^{-1/2} \bar{R}^{ab} \left( D_a ^{(3)}R\right) \left( D_b\lambda\right) 
- \lambda^{1/2} \left. \right. ^{(3)}R \left( D_a^{(3)}R\right)\left( D^a\lambda\right)
-\lambda^{1/2} \left[ ^{(3)}R\right]^3 \Biggr]~.
\eeq
Finally, adding Eq.~(\ref{eq:10.5}) and Eq.~(\ref{eq:12}) and using the contracted Bianchi identity
given in Eq.~(\ref{eq:contracted_bianchi}) demonstrates the desired result, Eq.~(\ref{eq:lem_proof}).
\end{proof}
%
%
\section{Generalised Lichnerowicz theorem for $4^{\rm th}$ order gravity}\label{sec:lich_thm}
Using the lemmas proved in the previous section, we will now extend the Lichnerowicz theorem to
$4^{\rm th}$ order gravity (for $\alpha \neq0$ and $\alpha \geq 3\beta$),
provided the spatial scalar curvature satisfies certain bounds.
\begin{thm}\label{thm:lich}
Consider a static space-time, with a spatial slice that is topologically $\mathbb{R}^3$,
 in which the spatial curvature everywhere
satisfies the following two conditions:
\beq\label{eq:restriction0}
m^2 - ^{(3)}R &\geq& 0~, \nonumber \\
\bar{R}^a_{~b}\bar{R}^b_{~a}\left( m^2 + {\cal R}\right) &\geq& 0~,
\eeq
where ${\cal R}$ is defined as,
\be
 {\cal R} \equiv \frac{ \bar{R}^a_{~b}\bar{R}^b_{~c}\bar{R}^c_{~a}}{ \bar{R}^a_{~b} \bar{R}^b_{~a}}~.
\ee
Let the space-time be asymptotically constant, unless both of the inequalities in Eq.~(\ref{eq:restriction0})
are saturated, in which case let the space-time be asymptotically flat.
Then for such a space-time the equations $R=0$ and 
\be
 \Box R_{ab} - m^2 R_{ab} + \frac{1}{2}g_{ab} R^{cd}R_{cd} - 2R^{cd}R_{acbd} =0~,
\ee
imply $R_{ab}=0$.
\end{thm}
\begin{proof}
The equation,
\be
 \Box R_{ab} - m^2 R_{ab} + \frac{1}{2}g_{ab} R^{cd}R_{cd} - 2R^{cd}R_{acbd} =0~,
\ee
implies,
\be\label{eq:13}
 I\equiv\left( \lambda^{1/2}\bar{R}^{ef} h^a_{~e}h^b_{~f} + \lambda^{-1/2}\left.\right.^{(3)}Rt^at^b  \right)\left(
 \Box R_{ab} - m^2 R_{ab} + \frac{1}{2}g_{ab} R^{cd}R_{cd} - 2R^{cd}R_{acbd}\right) =0~.
\ee
By Lemmas~(\ref{lem:2}) and (\ref{lem:3}) and Eq.~(\ref{eq:7}) and Eq.~(\ref{eq:contracted_bianchi})
of Lemma~(\ref{lem:1}) the integral of Eq.~(\ref{eq:13}) over the constant time slice, ${\cal S}$ is,
\beq\label{eq:14}
 && \int_{\cal S} \sqrt{h}{\rm d}^3 x\left( I\right) = \int_{\cal S} \sqrt{h} {\rm d}^3x \Biggl\{ D^a\left[ \lambda^{1/2}\left(
^{(3)}RD_a^{(3)}R + \bar{R}^{cd}D_a \bar{R}_{cd} - \bar{R}^{ef}D_e \bar{R}_{af} - ^{(3)}RD^b \bar{R}_{ba}\right)\right]
\nonumber \\
&&
-\lambda^{1/2} \left( D^a\bar{R}^{cd}\right)\left( D_a\bar{R}_{cd}\right)
-2\lambda^{1/2} \left( D_a\bar{R}_{bc}\right)\left( D^b \bar{R}^{ac}\right)
-\lambda^{1/2} \left[ D^a\left.\right.^{(3)}R - 2D^c\bar{R}^a_{~c}\right]\left[ D_a^{(3)} -2D^b\bar{R}_{ab}\right]
\nonumber \\
&&
-\lambda^{1/2}\left[\bar{R}_{cd}\bar{R}^{cd} \left( m^2+{\cal R}\right) + \left( ^{(3)}R\right)^2 \left( m^2
-^{(3)}R\right) \right] \Biggr\}=0~.
\eeq

Asymptotic constancy ensures that the boundary terms vanish, whilst the final two terms are negative definite
if the spatial curvature satisfies the properties,
\beq\label{eq:restrictions}
m^2 - ^{(3)}R &\geq& 0~, \nonumber \\
\bar{R}^a_{~b}\bar{R}^b_{~a}\left( m^2 + {\cal R}\right) &\geq& 0~.
\eeq
Finally, as shown in Appendix~(\ref{app}), the combination
\be\label{eq:pos_comb}
 \left( D^a\bar{R}^{cd}\right)\left( D_a\bar{R}_{cd}\right) + 2\left( D_a\bar{R}_{bc}\right)\left( D^b \bar{R}^{ac}\right)\geq0
\ee
for all $\bar{R}_{ab}$. Thus each term in the integrand of Eq.~(\ref{eq:14}) is negative definite and hence required to vanish.
Thus $D_cR_{ab}=0$ and if either of the inequalities in the expressions given in Eq.~(\ref{eq:restrictions}) fail to be saturated
in any open region, Eq.~(\ref{eq:14}) implies $R_{ab}=0$. If the expressions in Eq.~(\ref{eq:restrictions}) are saturated everywhere, then
asymptotic flatness is required to imply $R_{ab}=0$ everywhere.
\end{proof}

The equations of motion of $4^{\rm th}$ order gravity (with $\alpha -3\beta \geq 0$),
imply $R=0$ (by Theorem~(\ref{thm:p1})) and (from Eq.~(\ref{eq:trace_free}))
\be
 \hbar \alpha\left(  \Box R_{ab} - m^2 R_{ab} + \frac{1}{2}g_{ab} R^{cd}R_{cd} - 2R^{cd}R_{acbd}\right) =0~,
\ee
with $m^2 = \gamma \kappa^{-2} \hbar^{-1}\alpha^{-1}$. Thus, provided the spatial curvature everywhere obeys the bounds,
\beq
^{(3)}R &\leq& \frac{\gamma}{\hbar\alpha\kappa^2}  ~, \nonumber \\
{\cal R} &\geq& -\frac{\gamma}{\hbar\alpha\kappa^2}~,
\eeq
there are no, non-trivial, static, asymptotically constant vacuum solutions to $4^{\rm th}$ order gravity (with $\alpha -3\beta \geq0$
and $\alpha \neq 0$).

%
%
\section{Generalised Israel theorem for $4^{\rm th}$ order gravity}\label{sec:israel}
The result proved in Section~(\ref{sec:lich_thm}) demonstrates that there are no static, vacuum solutions
to $4^{\rm th}$ order gravity for space-times which have a spatial section that is topologically $\mathbb{R}^3$,
but one may ask whether a similar result holds for space-times in which the spatial slice contains a boundary.
In particular, in order to consider space-times with a (static) black-hole, we need to allow for a null, interior
boundary to our space-time i.e.\ the equal time hyper-surfaces become null.
At first sight the presence of boundary terms in Eq.~(\ref{eq:14}) would appear
to be a significant difficulty for such a space-time, however on the horizon, the normal to the spatial hyper-surface,
$n_a$ becomes null i.e.\ $n^an_a = \lambda =0$ and hence the boundary terms vanish.
Thus, if the space-time is asymptotically constant (which ensures the boundary terms at
infinity also vanish) and all interior boundaries are null, the result of Section~(\ref{sec:lich_thm}) will
still apply.
\begin{thm}\label{thm:israel}
Consider a static space-time, in which the spatial curvature everywhere satisfies the following two conditions:
\beq\label{eq:restrictions1}
^{(3)}R &\leq& \frac{\gamma}{\hbar\alpha\kappa^2}  ~, \nonumber \\
{\cal R} &\geq& -\frac{\gamma}{\hbar\alpha\kappa^2}~,
\eeq
Let the space-time be asymptotically constant, unless both of the inequalities in Eq.~(\ref{eq:restrictions1})
are saturated, in which case let the space-time be asymptotically flat. Let the spatial slice be bounded by a
null surface that is topologically $\mathbb{S}^2$.
Then the only solution to $4^{\rm th}$ order gravity (with $\alpha\neq 0$ and $\alpha - 3\beta \geq 0$),
 in the region exterior to the null surface, is $R_{ab}=0$.
\end{thm}
\begin{proof}
As in Theorem~\ref{thm:lich}, the equations of motion for $4^{\rm th}$ order gravity imply, (Eq.~(\ref{eq:14}))
\beq\label{eq:15}
 && \int \sqrt{h}{\rm d}^3 x\left( I\right) = \int \sqrt{h} {\rm d}^3x \Biggl\{ D^a\left[ \lambda^{1/2}\left(
^{(3)}RD_a^{(3)}R + \bar{R}^{cd}D_a \bar{R}_{cd} - \bar{R}^{ef}D_e \bar{R}_{af} - ^{(3)}RD^b \bar{R}_{ba}\right)\right]
\nonumber \\
&&
-\lambda^{1/2} \left( D^a\bar{R}^{cd}\right)\left( D_a\bar{R}_{cd}\right)
-2\lambda^{1/2} \left( D_a\bar{R}_{bc}\right)\left( D^b \bar{R}^{ac}\right)
-\lambda^{1/2} \left[ D^a\left.\right.^{(3)}R - 2D^c\bar{R}^a_{~c}\right]\left[ D_a^{(3)} -2D^b\bar{R}_{ab}\right]
\nonumber \\
&&
-\lambda^{1/2}\left[\bar{R}_{cd}\bar{R}^{cd} \left( \gamma \kappa^{-2} \alpha^{-1}+{\cal R}\right) + \left( ^{(3)}R\right)^2
\left( \gamma \kappa^{-2} \alpha^{-1} - ^{(3)}R\right) \right] \Biggr\}=0~.
\eeq

At infinity the boundary terms vanish because the spatial section is asymptotically constant, whilst on the interior
boundary they vanish because $\lambda=0$. Thus, just as in Theorem~(\ref{thm:lich}),
provided the $3$-dimensional curvature satisfies the conditions
 \beq
  \frac{\gamma}{ \hbar \alpha \kappa^{2}  } - ^{(3)}R &\geq& 0~, \nonumber \\
  \frac{\gamma}{ \hbar \alpha \kappa^{2} } + {\cal R} &\geq& 0~, \nonumber
 \eeq
each term in the integrand of Eq.~(\ref{eq:15}) is negative definite and hence vanishes. If either of these
inequalities is not saturated everywhere, asymptotic constancy is sufficient to imply $R_{ab}=0$, otherwise
asymptotic flatness is required.
\end{proof}

Since $4^{\rm th}$ order gravity implies $R_{ab}=0$ (for asymptotically constant and static space-times), provided
 \beq\label{eq:restrictions3}
 ^{(3)}R &\leq&  \frac{\gamma}{ \hbar \alpha \kappa^{2} }   ~, \nonumber \\
 {\cal R} &\geq& - \frac{\gamma}{ \hbar \alpha \kappa^{2}}~,
 \eeq
even in the presence of null boundaries, it follows that the only spherically symmetric solution is the (exterior of the)
Schwarzschild metric.

It is important to note here that this result relies on the existence of a null boundary
to the equal time hyper-surfaces and the existence of a vacuum ($T_{ab}=0$) everywhere between this horizon
and infinity. This is true for a (static, vacuum) space-time containing an event horizon, however the result
cannot be applied to the exterior of a general spherical object that does not contain such a horizon. In particular
this result to does not imply that the metric outside a spherically symmetric distribution of matter will be
Schwarzschild, unless that matter is entirely contained within the event horizon. This agrees with the (perturbative)
results of (for example)~\cite{Stelle:1977ry}.
%
%
\section{Non-vacuum solutions}\label{sec:non_vac}
In Sections~(\ref{sec:lich_thm}) and (\ref{sec:israel}) we showed that for $T_{ab}=0$, the static, asymptotically
constant solutions of GR are also solutions to $4^{\rm th}$ order gravity (with $\alpha \neq 0$ and $\alpha - 3\beta \geq 0$),
 at least in the case of space-times that have a spatial topology of $\mathbb{R}^3$ (Theorem~(\ref{thm:lich})) or space-times that
have internal null boundaries (Theorem~(\ref{thm:israel})), and whose spatial curvature
everywhere satisfies Eq.~(\ref{eq:restrictions3}). A natural question that one may ask is whether this
connection between the solutions of GR and $4^{\rm th}$ order gravity extends to the non-vacuum case i.e.\ $T_{ab}\neq 0$.
We will now demonstrate that it does not. In fact below we show that for, static, asymptotically flat space-times containing
a barotropic fluid with equation of state $\omega \geq -3^{-1/3}\approx -0.6933$, 
{\it only} the vacuum solutions of the two theories coincide, with all other solutions being distinct.
\begin{thm}
 Consider an asymptotically flat solution to GR, $^{\rm GR}g_{ab}$, with matter satisfying,
\be
 \int \sqrt{h} {\rm d}^3 x T^a_{~b}T^b_{~c}T^c_{~a} \geq 0~.
\ee 
Then $^{\rm GR}g_{ab}$ is a solution to $4^{\rm th}$ order gravity (with $\alpha \neq 0$ and $\alpha \neq  3\beta $)
iff $R_{ab}=0$.
\end{thm}
\begin{proof}
Consider a solution to the Einstein's equations:
\be\label{eq:proof_einstein}
 \gamma \kappa^{-2} \left( R_{ab} - \frac{1}{2}g_{ab}R\right) = \frac{1}{2} T_{ab}~,
\ee
that is static and asymptotically flat. Substituting this into the equations of motion for $4^{\rm th}$ order
gravity, Eq.~(\ref{eq:EoM}) we find,
\beq\label{eq:proof_EoM}
&&\left( \alpha - 2\beta\right) R_{;a;b} - \alpha \Box R_{ab} - \frac{1}{2}\left( \alpha -4\beta\right)
g_{ab}\Box R + 2\alpha R^{cd}R_{acbd} \nonumber \\
&&- 2\beta R R_{ab} - \frac{1}{2} g_{ab} \left( \alpha R^{cd}R_{cd}
- \beta R^2\right) = 0~.
\eeq
The trace of Eq.~(\ref{eq:proof_EoM}) is,
\be
 \left( \alpha - 3\beta\right) \Box R = 0~.
\ee
Then by Theorem~(\ref{thm:p1}), for $\alpha-3\beta \neq 0$, this implies $R=0$.
Thus Eq.~(\ref{eq:proof_EoM}) becomes,
\be
 -\alpha\left[ \Box R_{ab} +\frac{1}{2}g_{ab} R^{cd}R_{cd} - 2R^{cd}R_{acbd}\right]=0~,
\ee
which by Theorem~(\ref{thm:lich}) implies (see Eq.~(\ref{eq:14})),
\beq\label{eq:16}
 && \int \sqrt{h}{\rm d}^3 x\left( I\right) = \int \sqrt{h} {\rm d}^3x \Biggl\{ D^a\left[ \lambda^{1/2}\left(
^{(3)}RD_a^{(3)}R + \bar{R}^{cd}D_a \bar{R}_{cd} - \bar{R}^{ef}D_e \bar{R}_{af} - ^{(3)}RD^b \bar{R}_{ba}\right)\right]
\nonumber \\
&&
-\lambda^{1/2} \left( D^a\bar{R}^{cd}\right)\left( D_a\bar{R}_{cd}\right)
-2\lambda^{1/2} \left( D_a\bar{R}_{bc}\right)\left( D^b \bar{R}^{ac}\right)
-\lambda^{1/2} \left[ D^a\left.\right.^{(3)}R - 2D^c\bar{R}^a_{~c}\right]\left[ D_a^{(3)} -2D^b\bar{R}_{ab}\right]
\nonumber \\
&&
-\lambda^{1/2}\left[\bar{R}^a_{~b}\bar{R}^b_{~c}\bar{R}^c_{~a}  - \left( ^{(3)}R\right)^3 \right] \Biggr\}=0~.
\eeq
Thus, if
\be
 \int \sqrt{h}{\rm d}^3 x\left[\bar{R}^a_{~b}\bar{R}^b_{~c}\bar{R}^c_{~a}  - \left( ^{(3)}R\right)^3 \right] \geq 0~,
\ee
then each term in the integrand is negative definite and hence vanishes.

However it is easy to show that,
\beq\label{eq:R3_positive}
 R^a_{~b}R^b_{~c}R^c_{~a} &=& g^{ad}g^{be}g^{cf}R_{bd}R_{ce}R_{fa}~, \nonumber \\
&=& \bar{R}^a_{~b}\bar{R}^b_{~c}\bar{R}^c_{~a}  - \left( ^{(3)}R\right)^3~,
\eeq
where the second equality follows from decomposing $g^{ab}=h^{ab}-\frac{1}{\lambda} t^at^b$, using the Codacci and Gauss
equations (Eq.~(\ref{eq:Gauss}) and Eq.~(\ref{eq:Codazzi})) and noting that $R=0$.
Thus each term in the integrand of Eq.~(\ref{eq:16}) is negative definite, and
hence vanishes, provided,
\be
 \int \sqrt{h}{\rm d}^3 R^a_{~b}R^b_{~c}R^c_{~a} \geq 0~.
\ee
Using Eq.~(\ref{eq:proof_einstein}) and
the fact that $R=0$ for these solutions, this condition can be written as
\be\label{eq:cond}
 \int \sqrt{h} {\rm d}^3 x T^a_{~b}T^b_{~c}T^c_{~a}\geq 0~.
\ee

Thus, for matter satisfying Eq.~(\ref{eq:cond}), a static, asymptotically flat solution to GR is also
a solution to $4^{\rm th}$ order gravity {\it iff} $R_{ab}=0$, which,
by Theorems~(\ref{thm:lich}) and (\ref{thm:israel}) corresponds to the (static, asymptotically
flat) vacuum solutions in both GR and $4^{\rm th}$ order gravity.
Thus, other than the vacuum, there are no static, asymptotically flat solutions to GR, with matter satisfying Eq.~(\ref{eq:cond}),
that are also solutions to $4^{\rm th}$ order gravity.
\end{proof}

If we consider a perfect fluid the energy momentum tensor can be written as
$T^a_{~b}= \rho u^a u_b + P\left( \delta^a_{~b}  + u^a u_b\right)$, with $u_a$ a unit time-like vector tangent to the
observer's world-line, $P$ the fluid's pressure and $\rho$ its energy density. Eq.~(\ref{eq:cond}) then becomes,
\be
 \int \sqrt{h}{\rm d}^3x\left( \rho^3 + 3P^3\right)\geq 0~,
\ee
Assuming the strong energy condition holds, $\rho \geq 0$, then this condition is trivial for
all matter with positive pressure. If we consider a barotropic fluid with equation of state parameter
$\omega$ i.e.\  $P=\omega \rho$ then Eq.~(\ref{eq:cond}) reduces to,
\be
 \int \sqrt{h} {\rm d}^3 x \left( 1+3\omega^3 \right) \geq 0~.
\ee
This is, in particular, satisfied for $\omega \geq -3^{-1/3}\approx -0.6933$ everywhere, and hence 
for dust ($\omega=0$) and radiation ($\omega = 1/3$) fluids.

Thus for a static, asymptotically flat, space-time containing a single, barotropic fluid for which
$\omega \geq - 3^{-1/3}$ everywhere, no solutions to GR remain (static and asymptotically flat) solutions to 
$4^{\rm th}$ order gravity.

One can also consider, for example, the energy-momentum tensor of a Maxwell field,
\be
 T^a_{~b} = \frac{1}{4\pi}\left( F^a_{~c}F^c_{~b} - \frac{1}{4}g^a_{~b} F^d_{~e}F^e_{~d}\right)~,
\ee
with $F_{ab}$ the electromagnetic field tensor. For such a field, one can directly check that 
\be
 T^a_{~b}T^b_{~c}T^c_{~a} = \frac{1}{64\pi^3} \left[ \frac{1}{8} \left( {\rm Tr} {\bf F}^2\right)^3
- \frac{3}{4}\left( {\rm Tr} {\bf F}^4\right) \left( {\rm Tr} {\bf F}^2\right)
- {\rm Tr}{\bf F}^6\right] \geq 0~,
\ee
where the matrix ${\bf F}$ is given by $F^a_{~b}$ and hence is anti-symmetric and the final inequality
follows from direct computation, for a general anti-symmetric matrix.

%
%
\section{Conclusions}\label{sec:conc}
In the Sections~(\ref{sec:lich_thm}) and (\ref{sec:israel}) we demonstrated that there are no, non-trivial,
static, asymptotically constant solutions to $4^{\rm th}$ order gravity, {\it provided} the $3$-dimensional
curvature satisfies the following two conditions:
 \be
 ^{(3)}R  \leq  \frac{\gamma}{ \hbar\alpha \kappa^{2}} ~~~~{\rm and}~~~~ {\cal R} \geq - \frac{\gamma}{\hbar\alpha \kappa^{2}}~,
 \ee
everywhere outside a null horizon. The meaning of these conditions is the following. We require that the 
$3$-dimensional scalar curvature, measured by both $^{(3)}R$ and ${\cal R}$, to be everywhere smaller in magnitude than
the scale set by the ratio between $\gamma \kappa^{-2}$ and $\hbar\alpha$. This is exactly the scale at which corrections
to General Relativity will become significant. If these corrections are motivated by quantum corrections to
general relativity then this will be the quantum gravity scale. Thus our results can interpreted as follows:
there are no, non-trivial, static, asymptotically constant, vacuum solutions to $4^{\rm th}$ order gravity, so long as
the $3$-dimensional scalar curvature is everywhere less than the quantum gravity scale. In particular for
black holes, this implies that the Schwarzschild solution is the {\it unique}, spherically symmetric, static, asymptotically constant
vacuum solution to $4^{\rm th}$ order gravity, unless the $3$-dimensional scalar curvature exceeds the quantum
gravity scale {\it outside the horizon}, a condition that is easily achieved for macroscopic black-holes.

This also demonstrates the fact that the `no-hair' theorem applies to $4^{\rm th}$ order gravity provided that
the spatial curvature (outside the horizon) is everywhere less than the quantum gravity scale. Note however that
the converse does not follow i.e.\ if the spatial curvature outside the horizon is greater than the quantum
gravity scale, Theorem~(\ref{thm:israel}) {\it does not} imply that there are additional (static, asymptotically flat)
solutions.

Finally, in Section~(\ref{sec:non_vac}), we demonstrated that for matter satisfying the rather mild condition,
\be
 \int \sqrt{h} {\rm d}^3 x T^a_{~b}T^b_{~c}T^c_{~a}\geq 0~,
\ee
except for the vacuum, 
none of the static, asymptotically flat solutions to GR are solutions to $4^{\rm th}$ order gravity. In particular
this includes space-times containing single, barotropic fluids with equations of state everywhere satisfying, $\omega \geq - 3^{-1/3}$,
such as relativistic and non-relativistic Baryons and Fermions. A specific application of this is that
the GR solutions for (static, asymptotically flat) stars will fail to be solutions in $4^{\rm th}$ order gravity
(for the same matter distribution). This agrees with the results of~\cite{Stelle:1977ry}, which gives
the solutions to $4^{\rm th}$ order gravity for a particular, extended, spherically symmetric, system (calculated perturbatively
and close to the origin) which is shown not to be Schwarzschild.

\section*{Acknowledgments}

We would like to thank Abhay Ashtekar for extremely helpful discussions and for suggesting the topic.
This work was supported in part by NSF grants PHY0748336,
PHY0854743, The George A.\ and Margaret M.~Downsbrough Endowment and
the Eberly research funds of Penn State.

\section*{Appendix}\label{app}
In this appendix we prove various useful (if rather laborious) identities. First consider
$t^a\nabla_a t^b\nabla_b R_{cd}$, for a static space-time. Using the staticity condition,
\beq\label{eq:app_stat}
 {\cal L}_t R_{ab} &=&0~, \nonumber \\
\Rightarrow t^a\nabla_a R_{cd} &=& -2 R_{a(d} \nabla_{c)} t^a~,
\eeq
we can write,
\beq\label{eq:app1}
 t^a \nabla_a t^b \nabla_b R_{cd} &=& 2 \left [\left(\nabla_b t^a\right)R_{a(d}\left( \nabla_{c)}t^b\right)
   + R_{ab}\left( \nabla_{(d}t^a\right)\left( \nabla_{c)}t^b\right) \right] \nonumber \\
&& - R_{bd} t^a\nabla_a\nabla_c t^b - R_{cb} t^a\nabla_a \nabla_d t^b~.
\eeq
Using $t^c\nabla_c\lambda =0$ and hence that $\nabla_a\left( t^c\nabla_c\lambda\right)=0$, we find
\be
 t^c\nabla_c\nabla_a \lambda = -\left( \nabla_a t^c\right) \left( \nabla_c\lambda\right)~.
\ee
Thus,
\beq
 t^c\nabla_c\nabla_at^d &=& \frac{1}{2\lambda} \left( t^c t^d\nabla_c\nabla_a\lambda - t_a t^c\nabla_c\nabla^d\lambda\right)
\nonumber \\
&=& \frac{1}{2\lambda} \left( -t^d\left( \nabla_at^c\right) \left( \nabla_c\lambda\right) + t_a\left( \nabla^d t^c\right)
\left( \nabla_c\lambda\right) \right) =0~,
\eeq
where the final equality follows from expanding $\nabla_at_c = \frac{1}{\lambda} \left( t_{[c}\nabla_{a]}\lambda\right)$.
Using this in Eq.~(\ref{eq:app1}) gives,
\be\label{eq:app1.5}
 t^a\nabla_at^b\nabla_b R_{cd} = 2\left[ \left( \nabla_b t^a\right) \left( R_{a(d} \nabla_{c)} t^b\right) + R_{ab} 
\left( \nabla_c t^b\right)\left( \nabla_dt^a\right) \right]~,
\ee
which is Eq.~(\ref{eq:1}).

By using $\nabla_at_c = \frac{1}{\lambda} \left( t_{[c}\nabla_{a]}\lambda\right)$ and recalling that
 $t^a\nabla_a\lambda=0$, we find,
\beq
 \left( \nabla_b t^c\right)\left( \nabla_a t^d\right) &=& \frac{1}{4\lambda^2} \left( t^c\nabla_d\lambda - t_d\nabla^c\lambda\right)
\left( t^d\nabla_a\lambda - t_a\nabla^d\lambda\right)~, \nonumber \\
&=& \frac{-1}{4\lambda^2} t^c t_a\left( \nabla_d\lambda\right)\left( \nabla^d\lambda\right)
+\frac{1}{4\lambda} \left( \nabla^c\lambda\right)\left( \nabla_a\lambda\right)~,
\eeq
which is just Eq.~(\ref{eq:2}).

In order to show Eq.~(\ref{eq:3}) note that,
\be\label{eq:app2}
 \frac{1}{\lambda} t^c t^d \nabla_b R_{cd} = \nabla_b\left( \frac{1}{\lambda}t^ct^dR_{cd}\right) - R_{cd}\nabla
\left( \frac{1}{\lambda}t^ct^d\right)~.
\ee
Now we use the fact that for a static space time, the extrinsic curvature of the constant time hyper-surfaces vanishes
(this can be checked by direct computation). Thus the Gauss equation (Eq.~(\ref{eq:Gauss})) gives,
\be
 \frac{1}{\lambda} t^a t^b R_{ab}= ^{(3)}R-R~,
\ee
whilst the Codacci equation,(Eq.~(\ref{eq:Codazzi})) is simply $t^a h^b_{~c} R_{ab}=0$. A consequence of
the Codacci equation and $t^c\nabla_c \lambda =0$ is,
\be
 t^a\left( \nabla^b\lambda\right) R_{ab} = 0~,
\ee
Using this and the Gauss equation in Eq.~(\ref{eq:app2}) gives,
\be\label{eq:app3}
 \frac{1}{\lambda} t^c t^d\nabla_b R_{cd} = \nabla_b\left( ^{(3)}R - R\right)~.
\ee
With this we directly find Eq.~(\ref{eq:3}):
\be
 \frac{1}{2\lambda^2} h^{ab}\left( \nabla_a\lambda\right) t^c t^d\nabla_b R_{cd} = \frac{1}{2\lambda}\left( \nabla^a\lambda\right)
\left( \nabla_a\left( ^{(3)}R - R\right)\right)~.
\ee

Similarly, by noting,
\be
 \frac{1}{\lambda} t^c t^d h^{ab} \nabla_a h^e_{~b} \nabla_e R_{cd} = h^ab \nabla_a \left( \frac{1}{\lambda} t^c t^d h^e_{~b}
\nabla_e R_{cd}\right) - h^{ab}\left[ \nabla_a\left( \frac{1}{\lambda} t^c t^d\right)\right]\nabla_b R_{cd}~,
\ee
and using Eq.~(\ref{eq:app3}) we find,
\be
 \frac{1}{\lambda} t^c t^d h^{ab} \nabla_a h^e_{~b} \nabla_e R_{cd} = D_a D^a\left( ^{(3)}R -R\right)~,
\ee
which is Eq.~(\ref{eq:4}).

A key equation of Lemma~(\ref{lem:1}) is Eq.~(\ref{eq:3.5}), which we derive here. By decomposing $g_{ab} = h_{ab} - \frac{1}{\lambda}
t_at_b$ we have,
\be\label{eq:app5}
\Box R_{ab} - m^2 R_{ab} = h^{ab}\nabla_a\nabla_b R_{cd} - \frac{1}{\lambda} t^a t^b \nabla_a \nabla_bR_{cd} -m^2 R_{cd}=0~.
\ee
Splitting the projection $h^{ab} = h^{ac}h^b_{~c}$ and bringing one of the $3$-metrics, 
inside the covariant derivative the first term becomes
\be\label{eq:app6}
 h^{ab} \nabla_a \nabla_b R_{cd} = h^{ab} \nabla_a h^e_{~b} \nabla_e R_{cd}~,
\ee
where we used the fact that because $\nabla_a g_{bc}=0$, we have 
\be
 \nabla_a h_{bc} = \nabla_a\left( \frac{1}{\lambda} t_bt_c\right)~,
\ee
and $h^{ab}\nabla_a t_b=0$ by symmetry. For the second term in Eq.~(\ref{eq:app5}), we again
use the fact that $ \nabla_c g_{ab}=0$ to find,
\beq
 \nabla_c h_{ab} &=& \nabla_c\left(\frac{1}{\lambda}t_at_b\right)~, \nonumber \\
\Rightarrow t^a\nabla_c t_a &=& - t^a\nabla_a t_c = -\frac{1}{2} \nabla_c\lambda~,
\eeq
where the second line follows by expanding the right hand side, taking the trace over $\left( a,b\right)$ and 
noting that $\nabla_{(a} t_{b)}=0$. Using Eqs.~(\ref{eq:app1.5}) and (\ref{eq:app6}) in Eq.~(\ref{eq:app5}) gives,
\beq
&& h^{ab} \nabla_a h^e_{~b} \nabla_e R_{cd} + \frac{1}{2\lambda} h^{ab} \left( \nabla_a\lambda\right)
\left( \nabla_b R_{cd}\right) 
 \nonumber \\
&&+\frac{1}{4\lambda^3}\left[ t_c t^a\left( \nabla_b\lambda\right) \left( \nabla^b\lambda\right)
- \lambda\left( \nabla_c\lambda\right) \left(\nabla^a\lambda\right) \right] R_{ad}
\nonumber \\
&& - \frac{1}{2\lambda^2} \left[ t^a\left(\nabla_d\lambda\right) -t_d\left( \nabla^a\lambda\right)\right]
\left[ t^b\left( \nabla_c\lambda\right) - t_c \left( \nabla^d\lambda\right)\right] R_{ba}
\nonumber \\
&&+ \frac{1}{4\lambda^3} \left[ t^a t_d \left( \nabla_b\lambda\right)\left( \nabla^b\lambda\right)
-\lambda\left( \nabla^a\lambda\right)\left( \nabla_d\lambda\right) \right] R_{ca} - m^2R_{cd} =0~,
\eeq
which is Eq.~(\ref{eq:3.5}).

Finally, we need to show that the terms given in Eq.~(\ref{eq:pos_comb}) are positive i.e.\ that
\be
 \left( D_a \bar{R}_{bc} \right) \left( D^a \bar{R}^{bc}\right) +2\left( D_a\bar{R}_{bc}\right)
\left( D^b\bar{R}^{ac}\right) \geq 0~.
\ee
To do this, consider an arbitary tensor $T_{abc}$ that is symmetric in the final two indices
i.e. $T_{abc}=T_{acb}$. Now we use a tetrad basis, $e_a^i$, that is othonormal i.e. $e^i_a e^a_j = \delta^i_j$,
to decompose the (positive definite) metric as,
\be\label{eq:basis}
 h_{ab} = e^i_a e^j_b \delta_{ij}~.
\ee
With this one can write,
\beq
 T^{abc} T_{abc} &=& h^{ad}h^{de} h^{cf} T_{abc}T_{def} = T_{ijk} T_{mnp} \delta^{im} \delta^{jn} \delta^{kp}~,
\nonumber \\
 T^{abc} T_{bac} &=& h^{ae}h^{bd} h^{cf} T_{abc}T_{def} = T_{ijk} T_{mnp} \delta^{in} \delta^{jm} \delta^{kp}~.
\eeq
We now form the combination,
\be\label{eq:comb}
 T^{abc}T_{abc}+2T^{abc}T_{bac}~,
\ee
and expand out the Einstein summations. Since $T_{abc}$ is a three index tensor, the three possibilities
we need to consider are: the indices are the same, one is distinct and all are distinct. The first possibility clearly
involves only positive terms. In the orthonormal tetrad basis given by Eq.~(\ref{eq:basis}) these are
\be
  T^{abc}T_{abc}+2T^{abc}T_{bac} = 3\left( T_{111}T_{111} +T_{222}T_{222} + T_{333}T_{333}\right) +\dots~,
\ee
where the dots remind us that there are other terms we need to account for. The terms in Eq.~(\ref{eq:comb})
that arise from one index being distinct, with the other two being the same, can, by using Eq.~(\ref{eq:basis}),
be grouped to give,
\be\label{eq:group1}
  T^{abc}T_{abc}+2T^{abc}T_{bac} = \sum_{i\neq j} \left( 2 T_{iij} + T_{ijj} \right)\left(  2 T_{iij} + T_{ijj} \right)
+\dots~,
\ee
where there is no implicit summation over repeated indices and again the dots indicate that there
we are considering only certain terms. Finally we
need to account for the terms in Eq.~(\ref{eq:comb}) for which all the indices are distinct, which 
are given by,
\be\label{eq:group2}
  T^{abc}T_{abc}+2T^{abc}T_{bac} = \sum_{i\neq j\neq k} 2\left( T_{ijk}  + T_{jik} + T_{kij} \right)
\left( T_{ijk}  + T_{jik} + T_{kij} \right) + \dots~,
\ee
where again there is no implicit summation over repeated indices and the dots indicate that there
are additional terms. In Eq.~(\ref{eq:group1}) and Eq.~(\ref{eq:group2}), crucial use has been made
of the fact that $T_{ijk}=T_{ikj}$. Finally, putting these together we find that, in this orthonormal
tetrad basis,
\beq\label{eq:pos_final}
 T^{abc}T_{abc}+2T^{abc}T_{bac} &=& 3\sum_i T_{iii}T_{iii} + 
\sum_{i\neq j} \left( 2 T_{iij} + T_{ijj} \right)\left(  2 T_{iij} + T_{ijj} \right) \nonumber \\
&&+\sum_{i\neq j\neq k} 2\left( T_{ijk}  + T_{jik} + T_{kij} \right)
\left( T_{ijk}  + T_{jik} + T_{kij} \right)~,
\eeq
where once again there is no implicit summation over repeated indices. Clearly each term in Eq.~(\ref{eq:pos_final})
is positive. It is worth noting here that this result hold in any number of dimensions, provided the
metric, $h_{ab}$, is positive definite.

One may be concerned by the fact that if $T_{abc}$ were also anti-symmetric in the first two indices,
the combination given on the left-hand-side of Eq.~(\ref{eq:pos_final}), would appear to be negative. However there are no three index
tensors that are anti-symmetric in the first pair of indices and symmetric in the second pair, as can
easily be checked,
\be
 T_{abc} = - T_{bac} = -T_{bca} = T_{cba} = T_{cab} = - T_{acb} = - T_{abc}~.
\ee

Eq.~(\ref{eq:pos_final}) holds for any three index tensor that is symmetric in its final pair of indices
and in particular it holds for $T_{abc} = D_a \bar{R}_{bc}$, thus,
\be
 \left( D_a \bar{R}_{bc} \right) \left( D^a \bar{R}^{bc}\right) +2\left( D_a\bar{R}_{bc}\right)
\left( D^b\bar{R}^{ac}\right) \geq 0~,
\ee
as required.

\end{document}